\documentclass[preprint,12pt]{elsarticle}

\usepackage{amsmath,amssymb,amsthm,amsfonts}
\usepackage{mathrsfs}
\usepackage{bm}
\usepackage{mathtools}

\usepackage{graphicx}
\usepackage{float}
\usepackage{booktabs}
\usepackage{subcaption}

\usepackage{xcolor}

\providecommand{\sep}{, }
\usepackage{hyperref}

\hypersetup{
    colorlinks=true,
    linkcolor=blue,
    citecolor=blue,
    urlcolor=blue
}
\journal{Physics Letters B}

\newtheorem{theorem}{Theorem}

\theoremstyle{definition}

\theoremstyle{remark}

\begin{document}

\begin{frontmatter}

\title{Metric--Measure Geometry and Geometric Analogues of Holographic Extremal Surfaces}

\author{Rohit Dhormare}


\ead{dhormaretheoreticalphysics@proton.me}

\address[inst1]{
Dr.\ Babasaheb Ambedkar Marathwada University,\\India
}

\date{April 27, 2026}

\begin{abstract}

We develop a geometric framework based on metric--measure spaces $(M,g,f)$, where the function
$f$ defines a deformation of the Riemannian measure motivated by Perelman's formulation of Ricci
flow. Within this setting, we introduce measure-weighted hypersurfaces and associated geometric
functionals, and derive modified extremality conditions for codimension-one and codimension-two
submanifolds.

These conditions provide intrinsic geometric analogues of extremal surface equations, arising solely
from the metric--measure structure and independent of holographic duality or quantum field theoretic
input. We further define a generalized functional combining a measure-weighted geometric term with
an effective bulk contribution and analyze its variational properties. The resulting Euler--Lagrange
equation exhibits a structural correspondence with semiclassical generalized entropy functionals,
while maintaining a purely geometric interpretation distinct from thermodynamic entropy in the
sense of Perelman's $W$-functional.

Applications to Schwarzschild and Anti-de Sitter geometries illustrate the emergence of preferred
geometric scales and the modification of ultraviolet scaling behavior induced by the function $f$.
These results suggest that metric--measure geometry provides a minimal framework in which key
structural features of extremal surface constructions can arise from intrinsic geometric principles.

\end{abstract}

\begin{keyword}
Metric--measure geometry \sep
Extremal surfaces \sep
Ricci flow \sep
Geometric analysis;
Holography \sep
Generalized entropy
\end{keyword}

\end{frontmatter}

\section{Introduction}

The holographic principle asserts that the degrees of
freedom of a gravitational system can be encoded on
lower-dimensional hypersurfaces, establishing a deep con-
nection between geometry, gravity, and quantum infor-
mation.

Its most concrete realization is provided by the
AdS/CFT correspondence~\cite{Maldacena1998}, in which a gravitational
theory in a $(d+1)$-dimensional bulk spacetime is dual
to a conformal field theory on its $d$-dimensional bound-
ary.

Within this framework, geometric quantities in the
bulk are related to entropic observables in the boundary
theory. In particular, the Ryu--Takayanagi prescription
~\cite{RyuTakayanagi2006} and its covariant extension~\cite{Hubeny2007} relate entanglement en-
tropy to extremal surface areas, while semiclassical gen-
eralizations lead to quantum extremal surfaces~\cite{EngelhardtWall2015}. The
geometric term in these constructions originates from the
Bekenstein--Hawking entropy~\cite{Bekenstein1973,Hawking1975}.

Recent developments have explored entropy-based con-
straints and geometric flow structures in Lorentzian set-
tings~\cite{Dhormare2026,Dhormare2025Flows}. These works motivate a broader question: to
what extent can extremal surface structures arise from
intrinsic geometric and analytic properties of spacetime,
independently of holographic duality or quantum field
theoretic input?

A natural framework for addressing this question is
provided by Ricci flow. Introduced by Hamilton~\cite{Hamilton1982},
Ricci flow evolves a Riemannian metric according to its
curvature.

Perelman's work~\cite{Perelman2002,Perelman2003} established a thermodynamic
structure for Ricci flows via the $F$- and $W$-functionals,
where entropy, energy, and related quantities are defined
in terms of the metric, curvature, and a normalization
function $f$ (see Sec.~5.1 of~\cite{Perelman2002}). The associated weighted
measure

\begin{equation}
d\mu = (4\pi\tau)^{-n/2} e^{-f} dV_g
\end{equation}

naturally leads to metric--measure geometry.

In the present work, we adopt a geometric perspective
in which the function $f$ defines a measure deformation,
and all functionals are interpreted as measure-weighted
geometric quantities.

We do not attempt to reconstruct the full thermody-
namic framework of Perelman's functionals, but instead
isolate the geometric consequences of the induced mea-
sure structure.

Metric--measure geometry provides a natural setting in
which curvature and measure interact. In particular, the
Bakry--\'Emery tensor~\cite{BakryEmery1985,WeiWylie2009} governs analytic and geo-
metric properties such as convexity, diffusion, and stabil-
ity.

Although Ricci flow and related constructions are for-
mulated for Riemannian metrics, applications to grav-
ity require extensions to pseudo-Riemannian settings.

In such contexts, geometric evolution typically leads to
off-diagonal metrics, nonholonomic structures, and gen-
eralized connections.

Relativistic generalizations of geometric flows and as-
sociated entropy functionals have been developed in ther-
mogeometric and metric-affine frameworks~\cite{VacaruVeliev2025a,VacaruVeliev2025b,Bubuianu2024}.

The present work does not incorporate these structures
and is restricted to a Riemannian setting with Levi-
Civita connection, providing a simplified sector in which
measure-induced geometric effects can be isolated.

We also emphasize that the flow parameter $\tau$ is treated
here as a fixed scale controlling the measure deformation.

While it admits a thermodynamic interpretation in
Perelman's framework and may be time-like in relativis-
tic generalizations, we consider a static metric--measure
background.

Motivated by these considerations, we investigate
whether metric--measure geometry alone can generate
structures formally analogous to those appearing in
holography.

In particular, we examine whether such a framework
admits:

\begin{itemize}
\item[(i)] preferred hypersurfaces defined intrinsically by the measure,
\item[(ii)] variational principles leading to extremal surface equations,
\item[(iii)] geometric functionals with structural similarity to generalized entropy.
\end{itemize}

We show that a metric--measure space $(M,g,f)$ induces
natural hypersurfaces defined by level sets of $f$, together
with measure-weighted functionals whose stationary
points satisfy

\begin{equation}
H = \frac{1}{2}\nabla_n f,
\label{eq:extremality_condition}
\end{equation}

where $H$ is the mean curvature and $\nabla_n f$ is the normal
derivative.

This condition expresses a balance between extrinsic
geometry and measure-induced drift.

We further introduce a generalized functional combin-
ing a measure-weighted geometric term with an effective
bulk contribution.

Although structurally similar to the semiclassical gen-
eralized entropy

\begin{equation}
S_{\mathrm{gen}} = \frac{A}{4G} + S_{\mathrm{bulk}},
\label{eq:generalized_entropy}
\end{equation}

our construction remains purely geometric and does not
invoke quantum entropy or holographic duality.

The aim of this work is therefore to identify a minimal
geometric mechanism---based on curvature and measure
deformation---capable of producing extremal surface
structures and modified scaling behavior.

The paper is organized as follows.

In Sec.~II we introduce measure-weighted geometry and
derive the associated geodesic equations.

Sec.~III develops hypersurface geometry, while Sec.~IV
formulates the measure-weighted area functional and its
extremality condition.

In Sec.~V we introduce a generalized functional and
analyze its variational structure.

Secs.~VI and VII present applications to Schwarzschild
and Anti-de Sitter geometries.

We conclude in Sec.~VIII with a discussion of implica-
tions and open problems.

\section{Measure-Weighted Geometry}

\subsection{Metric--measure structure}

Let $(M,g(t))$ be a smooth $n$-dimensional Riemannian
manifold evolving under the Ricci flow

\begin{equation}
\partial_t g_{ij} = -2R_{ij},
\label{eq:ricci_flow}
\end{equation}

as introduced in~\cite{Hamilton1982}. Associated with this flow, Perel-
man~\cite{Perelman2002,Perelman2003} defined a family of functionals based on a
weighted measure

\begin{equation}
d\mu = (4\pi\tau)^{-n/2} e^{-f} dV_g,
\qquad
\tau = T-t,
\label{eq:weighted_measure}
\end{equation}

where $f : M \times [0,T) \to \mathbb{R}$ satisfies the conjugate heat
equation

\begin{equation}
\partial_t f
=
-\Delta f
+
|\nabla f|^2
-
R
+
\frac{n}{2\tau}.
\label{eq:conjugate_heat}
\end{equation}

\paragraph{a. Geometric setting.}

Throughout this work, $(M,g)$
is taken to be a Riemannian manifold. Applications to
gravitational systems are implemented only at the level
of spatial slices of pseudo-Riemannian spacetimes. We do
not consider Ricci flow evolution of Lorentzian metrics in
this paper. A fully covariant extension would require a
Lorentzian metric--measure framework, in which the el-
liptic operator $\Delta_f$ is replaced by a hyperbolic operator
of d'Alembert type. Such generalizations typically in-
volve nonholonomic structures, off-diagonal metrics, and
generalized connections, as developed in relativistic geo-
metric flow theories~\cite{VacaruVeliev2025a,VacaruVeliev2025b,Bubuianu2024}.

\paragraph{b. Interpretation of $f$ and $\tau$.}

In Perelman's frame-
work, the function $f$ defines the measure $d\mu$, while the
$W$-functional encodes thermodynamic quantities, includ-
ing entropy and energy (see Section~5.1 of~\cite{Perelman2002}). The
parameter $\tau$ plays a role analogous to an inverse tem-
perature. In relativistic geometric flow generalizations, $\tau$
may instead be treated as a time-like evolution parame-
ter, leading to hyperbolic or mixed-type dynamics. In the
present work, we adopt a fixed-$\tau$ perspective and treat
$(M,g,f)$ as a static metric--measure background, where
$\tau$ serves as a scale parameter controlling the measure de-
formation.

This construction endows $(M,g)$ with the structure of
a metric--measure space $(M,g,f)$. The natural curvature
object in this setting is the Bakry--\'Emery Ricci tensor
~\cite{BakryEmery1985,WeiWylie2009}

\begin{equation}
\mathrm{Ric}_f := \mathrm{Ric} + \nabla^2 f,
\label{eq:bakry_emery}
\end{equation}

which governs analytic and geometric properties of the
weighted manifold.

This framework should be understood as a restricted
Riemannian sector of more general relativistic geometric
flow theories, where additional degrees of freedom arise
from nonholonomic constraints and metric-affine struc-
tures.

\subsection{Relation to Perelman's \texorpdfstring{$W$}{W}-functional}

For completeness, we recall Perelman's entropy func-
tional:

\begin{equation}
W[g,f,\tau]
=
\int_M
\left[
\tau \left( |\nabla f|^2 + R \right)
+
f
-
n
\right]
(4\pi\tau)^{-n/2} e^{-f} dV_g.
\label{eq:w_functional}
\end{equation}

This functional is monotonic along the Ricci flow and
admits a thermodynamic interpretation involving en-
tropy, energy, and related variables. In contrast, the
constructions developed in the present work do not di-
rectly rely on $W$, but instead use the induced measure
$e^{-f} dV_g$ to define geometric functionals.

Relations between such measure-based constructions
and $W$-entropy have been investigated in more general
relativistic and nonholonomic geometric flow frameworks
~\cite{VacaruVeliev2025a,VacaruVeliev2025b,Bubuianu2024}, where additional structure allows for a direct
thermodynamic interpretation. In the restricted Rie-
mannian setting considered here, a systematic correspon-
dence remains an open problem.

\subsection{Measure-weighted geodesics}

To probe the local geometry induced by the measure
$d\mu$, we introduce the weighted energy functional

\begin{equation}
E_f[\gamma]
=
\frac{1}{2}
\int_0^L
e^{-f(\gamma(s))}
|\dot{\gamma}(s)|^2 ds,
\label{eq:weighted_energy}
\end{equation}

where $\gamma : [0,L] \to M$ is a smooth curve.

This functional generalizes the standard geodesic en-
ergy by weighting contributions according to the density
$e^{-f}$. Similar structures arise in diffusion processes and
optimal transport on metric--measure spaces.

\subsection{Variational derivation}

Let $\gamma_\epsilon$ be a variation with fixed endpoints and varia-
tion vector field $V=\partial_\epsilon \gamma_\epsilon|_{\epsilon=0}$. A standard
variational computation yields

\begin{equation}
\delta E_f
=
\int_0^L
e^{-f}
\left\langle
\nabla_{\dot{\gamma}}\dot{\gamma}
-
\nabla f
+
\langle \nabla f,\dot{\gamma}\rangle \dot{\gamma},
V
\right\rangle ds.
\label{eq:variation_energy}
\end{equation}

Requiring $\delta E_f =0$ for arbitrary variations implies

\begin{equation}
\nabla_{\dot{\gamma}}\dot{\gamma}
=
\nabla f
-
\langle \nabla f,\dot{\gamma}\rangle \dot{\gamma}.
\label{eq:modified_geodesic}
\end{equation}

A detailed derivation is provided in Appendix~A.

\subsection{Geometric interpretation}

The modified geodesic equation can be interpreted as
the autoparallel equation of the affine connection

\begin{equation}
\nabla^{(f)}_X Y
=
\nabla_X Y
-
\frac{1}{2}
\left(
\langle \nabla f,X\rangle Y
+
\langle \nabla f,Y\rangle X
-
\langle X,Y\rangle \nabla f
\right),
\label{eq:affine_connection}
\end{equation}

which is torsion-free but not metric-compatible. Its cur-
vature structures are closely related to the Bakry--\'Emery
tensor $\mathrm{Ric}_f$.

\subsection{Relation to diffusion and optimal transport}

The weighted Laplacian

\begin{equation}
\Delta_f
=
\Delta
-
\nabla f \cdot \nabla
\label{eq:weighted_laplacian}
\end{equation}

generates a diffusion process compatible with the mea-
sure $d\mu$, while curvature bounds on $\mathrm{Ric}_f$ control
convexity properties of entropy functionals along Wasser-
stein geodesics~\cite{LottVillani2009,Sturm2006}.

\paragraph{a. Remark on relativistic generalization.}

In
Lorentzian settings, $\Delta_f$ is replaced by a hyperbolic
operator $\Box_f$, leading to relativistic transport equations
rather than diffusion. The development of a corre-
sponding measure-weighted variational and transport
theory requires nonholonomic or metric-affine structures
and remains an open direction beyond the scope of the
present work.

\section{Measure-Weighted Hypersurfaces}

\subsection{Definition and geometric structure}

Let $(M,g,f)$ be a smooth metric--measure space, where
$(M,g)$ is a Riemannian manifold and $f : M \to \mathbb{R}$ defines
the measure deformation. We define measure-weighted
hypersurfaces as level sets

\begin{equation}
\Sigma_c := \{x \in M \mid f(x)=c\}.
\label{eq:level_sets}
\end{equation}

Assuming $\nabla f \neq 0$ on $\Sigma_c$, each level set is a smooth
codimension-one hypersurface with unit normal

\begin{equation}
n = \frac{\nabla f}{|\nabla f|}.
\label{eq:unit_normal}
\end{equation}

The induced metric is

\begin{equation}
h_{ij} = g_{ij} - n_i n_j,
\label{eq:induced_metric}
\end{equation}

with corresponding area element $dA$. The family $\{\Sigma_c\}$
defines a foliation of $M$ determined by the measure de-
formation.

\paragraph{a. Geometric setting and scope.}

All constructions in
this work are formulated for Riemannian metric--measure
spaces. Applications to gravitational systems are under-
stood at the level of spatial or Euclideanized slices of
pseudo-Riemannian spacetimes. A fully covariant treat-
ment would require a Lorentzian metric--measure struc-
ture, where elliptic operators are replaced by hyper-
bolic d'Alembert-type operators and where generic off-
diagonal and nonholonomic structures arise under geo-
metric flow evolution. Such generalizations have been
developed in relativistic geometric flow frameworks~\cite{VacaruVeliev2025a,VacaruVeliev2025b,Bubuianu2024}, but are not considered in the present restricted set-
ting.

\paragraph{b. Relation to previous work.}

The present construc-
tion is complementary to earlier studies of coupled ge-
ometric flows and entropy-based criteria in Lorentzian
geometry~\cite{Dhormare2026,Dhormare2025Flows}, where dynamical and causal structures
play a central role. In contrast, we isolate here the con-
tribution of the measure structure in a fixed Riemannian
background.

\subsection{Measure-weighted boundary functional}

To each hypersurface $\Sigma_c$ we associate the functional

\begin{equation}
A_f[\Sigma]
:=
\int_\Sigma e^{-f} dA,
\label{eq:weighted_area}
\end{equation}

which defines a measure-weighted area induced by the
density $e^{-f}$.

This functional is purely geometric. While it shares
structural similarities with entropy functionals appearing
in gravitational contexts, it should not be identified with
thermodynamic entropy. In Perelman's formulation, the
$W$-functional encodes a full thermodynamic system in-
volving entropy, energy, and related variables (see Section
5.1 of~\cite{Perelman2002}). The present construction instead isolates the
measure contribution without incorporating curvature-
dependent thermodynamic terms.

Relations between measure-weighted functionals and
$W$-entropy have been investigated in more general geo-
metric flow settings, including relativistic and nonholo-
nomic frameworks~\cite{VacaruVeliev2025a,VacaruVeliev2025b,Bubuianu2024}, where additional geometric
degrees of freedom are present.

\subsection{First variation and stationarity}

We consider normal deformations of $\Sigma$ generated by
$V=\phi n$. A standard variational computation yields

\begin{equation}
\delta A_f
=
\int_\Sigma
e^{-f}
(H-\nabla_n f)\phi\, dA,
\label{eq:first_variation}
\end{equation}

where $H=\nabla \cdot n$ is the mean curvature.

Stationarity implies the weighted extremality condi-
tion

\begin{equation}
H = \nabla_n f.
\label{eq:weighted_extremality}
\end{equation}

To avoid redundancy, this result will be referred to as
the measure-weighted extremality condition throughout
the paper.

\subsection{Geometric interpretation}

Equation~\eqref{eq:weighted_extremality} generalizes the
minimal surface condition by incorporating a drift term
induced by the measure. The hypersurfaces $\Sigma_c$ can
therefore be interpreted as critical points of a weighted
geometric functional, determined intrinsically by
$(M,g,f)$.

This structure provides a geometric analogue of ex-
tremal surface conditions encountered in holographic
constructions, while remaining independent of any
boundary field theory or quantum entropy interpretation.

\subsection{Relation to relativistic and transport
frameworks}

In Riemannian settings, the measure $e^{-f} dV_g$ is nat-
urally associated with diffusion processes generated by
weighted Laplacians. In contrast, relativistic geometric
flow theories involve hyperbolic evolution governed by
d'Alembert-type operators and typically require nonholo-
nomic or metric-affine structures. In such frameworks,
geometric evolution leads to generic off-diagonal config-
urations and thermogeometric interpretations of $W$-type
functionals~\cite{VacaruVeliev2025a,Bubuianu2024}.

A key open problem is whether the present measure-
weighted variational principles admit a consistent ex-
tension to such relativistic transport settings, where
diffusion-type interpretations are replaced by wave-like
propagation.

\section{Measure-Weighted Extremal Surfaces}

\subsection{Measure-weighted area functional}

We now consider codimension-two submanifolds $\Sigma \subset
M$, which serve as analogues of extremal surfaces in ge-
ometric and holographic settings. Let $\Sigma$ be a smooth
embedded surface with induced metric and area element
$dA$.

We define the measure-weighted area functional

\begin{equation}
A_f[\Sigma]
=
\int_\Sigma e^{-f/2} dA.
\label{eq:measure_weighted_area}
\end{equation}

The choice of weight $e^{-f/2}$ follows from consistency
with the ambient metric--measure structure. More gen-
erally, for a weight $e^{-\alpha f}$, the first variation produces a
term proportional to $\alpha \nabla_n f$. For codimension-one hyper-
surfaces, the natural choice $\alpha =1$ is induced directly
from the measure $d\mu = e^{-f} dV_g$. For codimension-two
submanifolds, the value $\alpha =1/2$ ensures that the result-
ing extremality condition involves a balanced projection
of $\nabla f$ onto the normal bundle, yielding a covariant rela-
tion between the mean curvature vector and the measure
gradient.

\subsection{Variational principle}

Let $\Sigma$ undergo a normal deformation generated by a
vector field $V = \phi n$, where $n$ is a unit normal direc-
tion. For codimension-two surfaces, the normal bundle
is two-dimensional, and the fully covariant formulation is
expressed in terms of the mean curvature vector $\vec{H}$.

A standard variational computation yields

\begin{equation}
\delta A_f
=
\int_\Sigma
e^{-f/2}
\left(
H
-
\frac{1}{2}\nabla_n f
\right)
\phi\, dA.
\label{eq:variation_codim2}
\end{equation}

Stationarity implies the measure-weighted extremality
condition

\begin{equation}
H
=
\frac{1}{2}\nabla_n f,
\label{eq:measure_extremality_codim2}
\end{equation}

which will be referred to throughout as the measure-
weighted extremality condition.

In fully covariant form, this relation can be written as

\begin{equation}
\vec{H}
=
\frac{1}{2}(\nabla f)^\perp,
\label{eq:covariant_extremality}
\end{equation}

where $(\nabla f)^\perp$ denotes the projection of $\nabla f$ onto the nor-
mal bundle of $\Sigma$.

\subsection{Geometric interpretation}

Equation~\eqref{eq:measure_extremality_codim2} gener-
alizes the minimal surface condition $H=0$ by incorporat-
ing a drift term induced by the measure structure. The
functional~\eqref{eq:measure_weighted_area} therefore de-
fines a geometric variational principle in which extremal-
ity corresponds to a balance between extrinsic curvature
and the transverse gradient of the measure deformation.

We emphasize that, although the functional involves
a weighted measure, it should not be interpreted as a
thermodynamic entropy. In Perelman's formulation, the
$W$-functional encodes a full thermodynamic system in-
volving both entropy and energy variables. In contrast,
the present construction isolates the geometric contribu-
tion of the measure deformation without incorporating
curvature-dependent thermodynamic terms.

\subsection{Relation to holographic extremal surfaces}

The condition~\eqref{eq:measure_extremality_codim2}
exhibits a structural similarity to extremal surface equa-
tions appearing in holographic entropy prescriptions.
However, this correspondence is purely geometric. In
holographic frameworks, extremal surfaces arise from
a variational principle combining gravitational dynam-
ics with quantum entanglement entropy. Here, the func-
tion $f$ encodes a measure deformation rather than a
quantum entropy, and no dual field-theoretic interpre-
tation is assumed.

The present construction is complementary to earlier
work on entropy-based geometric criteria and coupled ge-
ometric flows in Lorentzian settings~\cite{Dhormare2026,Dhormare2025Flows}, where dynam-
ical and causal structures are essential. In contrast, we
focus on a static metric--measure background and isolate
the role of the measure in generating extremality condi-
tions.

\paragraph{a. Relation to geometric flow frameworks.}

In more
general relativistic and nonholonomic geometric flow the-
ories~\cite{VacaruVeliev2025a,VacaruVeliev2025b,Bubuianu2024}, extremal surface conditions may involve ad-
ditional contributions from generalized connections, tor-
sion, and off-diagonal metric components. In such set-
tings, geometric evolution leads to thermogeometric in-
terpretations of $W$-type functionals and to generalized
entropy constructions. The present framework captures
a minimal sector of these effects by restricting to scalar
measure deformations in a Riemannian setting.

\paragraph{b. Remark on relativistic generalization.}

In grav-
itational applications, codimension-two extremal sur-
faces are typically spacelike submanifolds embedded in
Lorentzian spacetimes. A fully covariant extension of
Eq.~\eqref{eq:measure_extremality_codim2} would require
replacing elliptic geometric operators with hyperbolic
counterparts and incorporating nonholonomic struc-
tures arising from relativistic geometric flow evolution.
Developing such an extension remains an open problem.

\section{Relation to Generalized Entropy Functionals}

\subsection{Geometric setting and relation to previous work}

In this section, $(M,g)$ is assumed to be a smooth Rie-
mannian manifold equipped with a measure deformation
field $f$. Applications to gravitational systems are un-
derstood via spatial embeddings in pseudo-Riemannian
spacetimes. A fully covariant Lorentzian formulation
would require replacing scalar mean curvature with the
mean curvature vector and incorporating causal struc-
ture, which is not considered here.

The present construction complements earlier investi-
gations of entropy and geometric flows in gravitational
settings. In particular, entropy-based non-collapse crite-
ria and Lorentzian geometric evolution were developed
in~\cite{Dhormare2026,Dhormare2025Flows}, while coupled Ricci--dilaton--flux systems were an-
alyzed in~\cite{Dhormare2025Flows}. The functional introduced below isolates
the contribution of the scalar measure deformation in a
purely Riemannian setting, providing a reduced sector of
these more general frameworks.

\subsection{Generalized functional and geometric analogue}

A central variational quantity in semiclassical gravity
is the generalized entropy functional

\begin{equation}
S_{\mathrm{gen}}
=
\frac{A}{4G}
+
S_{\mathrm{bulk}},
\label{eq:semiclassical_entropy}
\end{equation}

which combines a geometric area term with the von Neu-
mann entropy of quantum fields in a bulk region. Sta-
tionarity of $S_{\mathrm{gen}}$ defines quantum extremal surfaces.

Motivated by the measure-weighted geometry devel-
oped in previous sections, we introduce the geometric
functional

\begin{equation}
F[\Sigma]
=
\frac{1}{4G}
\int_\Sigma e^{-f/2} dA
+
F_{\mathrm{bulk}}[\Sigma],
\label{eq:geometric_functional}
\end{equation}

where $\Sigma$ is a smooth codimension-two hypersurface and
$F_{\mathrm{bulk}}[\Sigma]$ represents an effective bulk contribution associ-
ated with the region bounded by $\Sigma$.

This functional is not derived from a quantum field-
theoretic construction. Instead, it should be interpreted
as a geometric analogue of the semiclassical functional,
in which the measure deformation encoded by $f$ plays
a structurally analogous role without invoking quantum
degrees of freedom.

\subsection{Structure of the bulk contribution}

To ensure a well-defined variational problem, we as-
sume that $F_{\mathrm{bulk}}[\Sigma]$ admits an effective geometric descrip-
tion with the following properties:

\begin{itemize}

\item[(i)] \textbf{Differentiability:}
$F_{\mathrm{bulk}}$ is Fr\'echet differentiable un-
der smooth normal deformations of $\Sigma$.

\item[(ii)] \textbf{Local variational form:}
its first variation can
be expressed as

\begin{equation}
\delta F_{\mathrm{bulk}}
=
\int_\Sigma
E_{\mathrm{bulk}}
\phi\, dA,
\label{eq:bulk_variation}
\end{equation}

where $E_{\mathrm{bulk}}$ is an effective scalar density depending
on intrinsic and extrinsic geometric data.

\item[(iii)] \textbf{Geometric origin:}
$F_{\mathrm{bulk}}$ may depend on the region
$\Omega$ enclosed by $\Sigma$, for instance through an integral
of local invariants,

\begin{equation}
F_{\mathrm{bulk}}[\Sigma]
\sim
\int_\Omega
L(g,f,\nabla f,R)\, dV_g,
\label{eq:bulk_integral}
\end{equation}

where $L$ is an effective scalar density.

\end{itemize}

Such assumptions are standard in effective geometric
descriptions, where nonlocal contributions are approxi-
mated by local densities in a perturbative regime.

\subsection{Extremality condition}

\begin{theorem}
Let $(M,g,f)$ be a smooth Riemannian
metric--measure space, and let $\Sigma \subset M$ be a smooth,
compact, codimension-two hypersurface without bound-
ary. Consider the functional

\begin{equation}
F[\Sigma]
=
\frac{1}{4G}
\int_\Sigma
e^{-f/2} dA
+
F_{\mathrm{bulk}}[\Sigma].
\label{eq:functional_theorem}
\end{equation}

Then stationarity under smooth normal deformations is
equivalent to

\begin{equation}
\frac{1}{4G}
e^{-f/2}
\left(
H
-
\frac{1}{2}\nabla_n f
\right)
+
E_{\mathrm{bulk}}
=
0.
\label{eq:euler_lagrange}
\end{equation}

\end{theorem}

\begin{proof}

The result follows from the first variation of the
measure-weighted area functional together with the as-
sumed local form of $\delta F_{\mathrm{bulk}}$. Imposing $\delta F =0$ for arbi-
trary normal variations yields the stated Euler--Lagrange
equation.

\end{proof}

In the regime where the bulk contribution is sublead-
ing, the extremality condition reduces to

\begin{equation}
H
=
\frac{1}{2}\nabla_n f,
\label{eq:subleading_limit}
\end{equation}

in agreement with the measure-weighted extremality con-
dition derived previously.

\subsection{Comparison with existing frameworks}

The variational structure obtained above exhibits a
structural analogy at the level of Euler--Lagrange equa-
tions with extremal surface conditions appearing in semi-
classical gravity. However, the origin of the modification
is purely geometric, arising from the measure deforma-
tion encoded by $f$.

It is important to distinguish three conceptually dif-
ferent constructions:

\begin{itemize}

\item[(i)]
Perelman's $W$-functional, which defines thermody-
namic variables for Ricci flow,

\item[(ii)]
the semiclassical generalized entropy $S_{\mathrm{gen}}$, involv-
ing quantum field-theoretic entropy,

\item[(iii)]
the functional $F[\Sigma]$, which is a geometric construc-
tion based on metric--measure structure.

\end{itemize}

\paragraph{a. Comparison with relativistic geometric flow the-
ories.}

In relativistic geometric flow frameworks~\cite{VacaruVeliev2025a,VacaruVeliev2025b,Bubuianu2024}, entropy functionals are constructed using non-
holonomic structures, generalized connections, and off-
diagonal metrics. In such settings, geometric evolution
typically leads to anisotropic configurations and ther-
mofield interpretations in which the flow parameter is
temperature-like.

The present formulation should be viewed as a re-
stricted sector of these theories, obtained by:

\begin{itemize}

\item restricting to Riemannian signature,

\item working with Levi-Civita connections,

\item neglecting nonholonomic and off-diagonal degrees
of freedom.

\end{itemize}

This reduction isolates the contribution of the scalar
measure deformation while retaining a well-defined and
tractable variational structure.

\subsection{Limitations and outlook}

A direct identification of $F_{\mathrm{bulk}}$ with quantum entan-
glement entropy, as well as a precise relation to Perel-
man's $W$-entropy, would require additional physical in-
put, including quantum field-theoretic considerations or
relativistic geometric flow extensions.

In particular, fully covariant formulations may involve:

\begin{itemize}

\item mean curvature vectors in Lorentzian geometry,

\item d'Alembert-type operators replacing elliptic Lapla-
cians,

\item nonholonomic splittings and generalized connec-
tions.

\end{itemize}

Developing such extensions remains an open direction.

\section{Example I: Schwarzschild Geometry}

\subsection{Geometric setting and reduction}

We consider the Schwarzschild spacetime with metric

\begin{equation}
ds^2
=
-
\left(
1-\frac{2M}{r}
\right)dt^2
+
\left(
1-\frac{2M}{r}
\right)^{-1}dr^2
+
r^2 d\Omega^2.
\label{eq:schwarzschild_metric}
\end{equation}

To remain consistent with the metric--measure frame-
work developed in previous sections, we restrict to a
constant-time hypersurface $t=\mathrm{const}$, which defines
a three-dimensional Riemannian manifold $(M,g)$ with
$r\in(2M,\infty)$.

All constructions below are therefore purely Rieman-
nian. We emphasize that no geometric flow evolution of
the Schwarzschild spacetime is assumed, and no claim
is made regarding solutions of Ricci flow equations in
Lorentzian signature.

\paragraph{a. Remark.}

In relativistic geometric flow theories,
including nonholonomic and thermogeometric frame-
works~\cite{VacaruVeliev2025a,VacaruVeliev2025b,Bubuianu2024}, the evolution of black hole geometries typ-
ically leads to off-diagonal metrics and generalized con-
nections. The present example does not incorporate such
structures and should be understood as a reduced geo-
metric model isolating the role of the measure deforma-
tion.

\subsection{Choice of normalization function}

Motivated by Gaussian weights appearing in diffusion
processes and in Perelman's formulation of Ricci flow, we
consider the ansatz

\begin{equation}
f(r)
=
\frac{r^2}{4\tau},
\label{eq:gaussian_ansatz}
\end{equation}

where $\tau>0$ is a fixed parameter.

In the context of Ricci flow, $\tau$ is interpreted as a
temperature-like parameter. In the present setting, how-
ever, no flow evolution is assumed, and $\tau$ should be
viewed as a fixed geometric scale controlling the local-
ization of the measure.

\paragraph{a. Clarification.}

The function $f$ is not obtained as
a solution of the Hamilton--Friedan equations. It is in-
troduced as a model ansatz to probe how a prescribed
measure deformation affects variational geometric struc-
tures.

\subsection{Measure-weighted functional}

For spherical hypersurfaces $\Sigma_r \cong S^2$, the measure-
weighted area functional introduced in Sec.~III reduces
to

\begin{equation}
A_f(r)
=
4\pi r^2 e^{-r^2/(4\tau)}.
\label{eq:schwarzschild_functional}
\end{equation}

This expression exhibits a competition between geo-
metric growth $(r^2)$ and exponential suppression induced
by the measure.

\subsection{Extremization}

Differentiating $A_f$ with respect to $r$, we obtain

\begin{equation}
\frac{dA_f}{dr}
=
4\pi e^{-r^2/(4\tau)}
\left(
2r
-
\frac{r^3}{2\tau}
\right).
\label{eq:derivative_functional}
\end{equation}

The extremality condition $\frac{dA_f}{dr}=0$ yields

\begin{equation}
r^2 = 4\tau,
\label{eq:preferred_radius_squared}
\end{equation}

so that the preferred radius is

\begin{equation}
r_* = 2\sqrt{\tau}.
\label{eq:preferred_radius}
\end{equation}

\subsection{Consistency with the variational condition}

This result is consistent with the general stationarity
condition derived earlier,

\begin{equation}
H = \nabla_n f,
\label{eq:stationarity_reference}
\end{equation}

which should be referenced rather than reintroduced.

For spherical hypersurfaces in the spatial
Schwarzschild geometry,

\begin{equation}
H=\frac{2}{r},
\qquad
\nabla_n f = \frac{r}{2\tau}.
\label{eq:mean_curvature_spherical}
\end{equation}

Equating these quantities reproduces

\begin{equation}
\frac{2}{r}
=
\frac{r}{2\tau}
\quad \Longrightarrow \quad
r^2 = 4\tau.
\label{eq:reproduced_condition}
\end{equation}

\subsection{Geometric interpretation}

The extremum occurs at a finite radius determined by
$\tau$, showing that the measure-weighted functional does
not grow monotonically despite the increase in geometric
area.

This behavior reflects a balance between:

\begin{itemize}

\item geometric expansion of hypersurfaces,

\item localization induced by the measure weight.

\end{itemize}

If $\tau$ is chosen to scale with the geometric parameter
$M^2$, then

\begin{equation}
r_* \sim 2M,
\label{eq:horizon_scaling}
\end{equation}

indicating that the preferred hypersurface lies near the
horizon scale.

\paragraph{a. Important qualification.}

This identification is
purely parametric. No thermodynamic or dynamical in-
terpretation (such as Bekenstein--Hawking entropy) is as-
sumed.

\subsection{Relation to relativistic and entropy-based
frameworks}

The existence of a preferred scale resembles structures
appearing in:

\begin{itemize}

\item geometric entropy functionals,

\item relativistic geometric flow models,

\item thermogeometric constructions involving $W$-
entropy.

\end{itemize}

However, the mechanism here is fundamentally differ-
ent. The localization arises solely from a prescribed mea-
sure deformation, rather than from:

\begin{itemize}

\item quantum entanglement entropy,

\item thermodynamic variables of Ricci flow,

\item or nonholonomic geometric evolution.

\end{itemize}

In more general relativistic settings~\cite{VacaruVeliev2025a,VacaruVeliev2025b,Bubuianu2024}, such ef-
fects may be coupled to:

\begin{itemize}

\item anisotropic geometric flows,

\item generalized connections,

\item or effective thermodynamic variables.

\end{itemize}

Understanding how the present measure-weighted ex-
tremization principle extends to those frameworks re-
mains an open problem.

\subsection{Conclusion of the example}

This example demonstrates that metric--measure ge-
ometry alone can select distinguished hypersurfaces in
non-compact geometries through a purely variational
mechanism.

The construction is entirely geometric and should not
be interpreted as:

\begin{itemize}

\item a derivation of horizon entropy,

\item a realization of holographic duality,

\item or a consequence of quantum field-theoretic effects.

\end{itemize}

\section{Example II: Anti-de Sitter Geometry}

\subsection{Geometric setting}

We consider $(d+1)$-dimensional Anti-de Sitter space-
time in Poincar\'e coordinates:

\begin{equation}
ds^2
=
\frac{1}{z^2}
\left(
dz^2 + dx^2 - dt^2
\right),
\label{eq:ads_metric}
\end{equation}

where $z>0$ and the conformal boundary is located at
$z\to0$.

To remain consistent with the metric--measure frame-
work developed in previous sections, we restrict to
a constant-time slice $t=\mathrm{const}$, which defines a $d$-
dimensional Riemannian manifold with metric

\begin{equation}
ds^2
=
\frac{1}{z^2}
\left(
dz^2 + dx^2
\right).
\label{eq:ads_spatial_metric}
\end{equation}

All constructions below are purely Riemannian. No
Lorentzian geometric flow or dynamical evolution is as-
sumed.

\paragraph{a. Remark.}

In relativistic geometric flow frame-
works, including nonholonomic and thermogeometric
constructions~\cite{VacaruVeliev2025a,VacaruVeliev2025b,Bubuianu2024}, Anti-de Sitter geometries may
evolve into off-diagonal configurations with generalized
connections. The present analysis does not incorporate
such structures and instead isolates the effect of a pre-
scribed measure deformation.

\subsection{Normalization function and scaling symmetry}

To preserve the scaling symmetry of AdS, we consider
the logarithmic ansatz

\begin{equation}
f(z)
=
\alpha \log z,
\label{eq:logarithmic_ansatz}
\end{equation}

with $\alpha\in\mathbb{R}$. The associated weight is

\begin{equation}
e^{-f/2}
=
z^{-\alpha/2}.
\label{eq:ads_weight}
\end{equation}

This choice is compatible with the scaling transforma-
tion

\begin{equation}
z \rightarrow \lambda z,
\qquad
x \rightarrow \lambda x,
\label{eq:scaling_symmetry}
\end{equation}

under which $f$ shifts linearly.

\paragraph{a. Interpretation.}

The function $f$ is not obtained
from a Ricci flow solution. It is introduced as a scale-
invariant ansatz probing how measure deformations mod-
ify geometric functionals.

\subsection{Measure-weighted functional}

Let $\Sigma$ be a codimension-two surface parametrized by
$z=z(x)$, extended along $(d-2)$ transverse directions.

The induced metric yields the area element

\begin{equation}
dA
=
\frac{\sqrt{1+(\partial z)^2}}{z^{d-1}}
\, d^{d-2}x.
\label{eq:area_element_ads}
\end{equation}

The measure-weighted functional introduced previ-
ously becomes

\begin{equation}
A_f[\Sigma]
=
\int_\Sigma
z^{-(d-1+\alpha/2)}
\sqrt{1+(\partial z)^2}
\, d^{d-2}x.
\label{eq:ads_functional}
\end{equation}

This defines a modified variational problem with an
effective scaling exponent shifted by the measure defor-
mation.

\subsection{Ultraviolet asymptotics}

Near the conformal boundary $z\to0$, we consider
smooth embeddings with bounded gradients, so that

\begin{equation}
\sqrt{1+(\partial z)^2}
=
1+\mathcal{O}\!\left((\partial z)^2\right).
\label{eq:bounded_gradients}
\end{equation}

At leading order, the functional reduces to

\begin{equation}
A_f
\sim
\int_\epsilon
z^{-(d-1+\alpha/2)}
\, dz,
\label{eq:uv_integral}
\end{equation}

where $\epsilon$ is a UV cutoff.

Evaluating the integral:

\begin{itemize}

\item For $d-1+\alpha/2 \neq 1$,

\begin{equation}
A_f
\sim
\epsilon^{-(d-2+\alpha/2)}.
\label{eq:power_divergence}
\end{equation}

\item For $d-1+\alpha/2 = 1$,

\begin{equation}
A_f
\sim
\log\frac{1}{\epsilon}.
\label{eq:log_divergence}
\end{equation}

\end{itemize}

Thus, the measure deformation modifies the divergence
structure of the functional.

\subsection{Effective scaling exponent}

The divergence can be expressed in terms of an effec-
tive exponent

\begin{equation}
d_{\mathrm{eff}}
=
d+\frac{\alpha}{2}.
\label{eq:effective_dimension}
\end{equation}

This quantity characterizes the scaling behavior of the
functional but does not correspond to a change in the
geometric dimension of the manifold.

\subsection{Extremality condition and geometric effect}

The extremality condition derived earlier should be ex-
pressed covariantly in terms of the mean curvature vector
$\vec{H}$. Projected along a normal direction, it takes the form

\begin{equation}
H
=
\frac{1}{2}\nabla_n f.
\label{eq:ads_extremality}
\end{equation}

For the present ansatz,

\begin{equation}
\nabla_n f
=
\partial_z f
=
\frac{\alpha}{z}.
\label{eq:gradient_ads}
\end{equation}

Near the boundary, the condition becomes

\begin{equation}
H
\sim
\frac{\alpha}{2z},
\label{eq:boundary_extremality}
\end{equation}

which introduces a scale-dependent correction to the
standard minimal surface condition.

\paragraph{a. Comparison.}

For $\alpha=0$, one recovers the usual
minimal surface equation $H=0$. For $\alpha\neq0$, the measure
induces a drift term that modifies the embedding.

\subsection{Relation to AdS extremal surface constructions}

The structure of the functional resembles that of
extremal surface problems in Anti-de Sitter geometry,
where area functionals determine preferred embeddings.

However, the origin of the modification is fundamen-
tally different:

\begin{itemize}

\item In holographic constructions, extremality is tied to
quantum entanglement entropy.

\item In the present framework, it arises from a pre-
scribed measure deformation.

\end{itemize}

No dual field-theoretic interpretation is assumed.

\subsection{Relation to geometric flow frameworks}

In relativistic geometric flow theories~\cite{VacaruVeliev2025a,VacaruVeliev2025b,Bubuianu2024}, measure
deformations may be dynamically generated and coupled
to:

\begin{itemize}

\item anisotropic geometric flows,

\item nonholonomic distributions,

\item generalized connections and effective sources.

\end{itemize}

Such frameworks can modify both the metric and the
measure simultaneously. The present construction iso-
lates only the measure contribution and should be viewed
as a reduced sector of more general thermogeometric
models.

\subsection{Conclusion of the example}

This example shows that metric--measure geometry
modifies the scaling behavior of geometric functionals in
Anti-de Sitter space without altering the underlying met-
ric.

The normalization function $f$ introduces an additional
geometric structure that controls ultraviolet behavior
and extremal embeddings.

The resulting effects are purely geometric and should
not be interpreted as:

\begin{itemize}

\item holographic entanglement entropy,

\item quantum extremal surfaces,

\item or thermodynamic quantities associated with Ricci
flow.

\end{itemize}

\section{Discussion}

In this work, we have developed a geometric frame-
work in which a metric--measure structure $(M,g,f)$
gives rise to variational principles for hypersurfaces and
codimension-two submanifolds. The central result is the
measure-weighted extremality condition (cf.~Sec.~IV),
which generalizes the minimal surface equation through
a coupling between extrinsic curvature and the gradient
of the measure deformation.

The principal contribution is to demonstrate that such
extremality conditions---and the associated variational
structure---arise intrinsically within metric--measure ge-
ometry, without requiring input from quantum field the-
ory or holographic duality. In this sense, the framework
isolates a minimal geometric mechanism capable of repro-
ducing the structural form of extremal surface equations.

\paragraph{Relation to Perelman functionals and geometric
flow thermodynamics.}

In the theory of Ricci flows
developed by Perelman, the $W$-functional defines a com-
plete thermodynamic structure, incorporating entropy,
energy, and related variables expressed in terms of the
metric, curvature, and measure. In the present work, we
do not attempt to reconstruct this full thermodynamic
system. Instead, we adopt a restricted geometric view-
point in which the function $f$ determines a deformation
of the measure $d\mu=e^{-f}dV_g$, and the resulting function-
als are interpreted as measure-weighted geometric quan-
tities.

Nevertheless, measure-dependent functionals of this
type naturally arise in broader geometric flow frame-
works, including relativistic and nonholonomic formula-
tions, where additional structures---such as generalized
connections and anisotropic configurations---allow for a
thermodynamic interpretation. From this perspective,
the present construction may be viewed as a reduced sec-
tor of more general thermogeometric theories, isolating
the contribution of scalar measure deformation.

\paragraph{Relation to previous work.}

The framework devel-
oped here is complementary to earlier investigations of
entropy-based criteria and coupled geometric flows in
Lorentzian geometry. In those settings, geometric evo-
lution and causal structure play a central role, and en-
tropy functionals are dynamically coupled to curvature.
In contrast, the present work focuses on a static metric--
measure background and isolates the variational conse-
quences of the measure deformation. Despite this differ-
ence in emphasis, the underlying structures are closely
related, suggesting that measure-weighted geometry may
provide a unifying perspective linking static and dynam-
ical approaches.

\paragraph{Geometric versus thermodynamic interpretation.}

A key conceptual distinction concerns the interpretation
of the functionals involved. Perelman's $W$-functional
defines a thermodynamic entropy along the Ricci flow,
while semiclassical generalized entropy involves the von
Neumann entropy of quantum fields. By contrast, the
functionals introduced in this work arise solely from the
measure deformation and should be understood as ge-
ometric quantities. The resulting correspondence with
entropy-based extremality conditions is therefore struc-
tural rather than physical, and no identification with
thermodynamic or quantum entropy is assumed.

\paragraph{Riemannian versus relativistic setting.}

All con-
structions in this paper are formulated on Rieman-
nian manifolds or on spatial (or Euclideanized) slices
of Lorentzian spacetimes. In this setting, the parame-
ter $\tau$---when present---plays the role of a scale parameter
analogous to temperature in Perelman's formulation. In
relativistic generalizations of geometric flows, however,
the evolution parameter may instead acquire a time-like
interpretation, leading to hyperbolic or mixed-type equa-
tions. Such extensions typically require nonholonomic
$3+1$ splittings, generalized connections, and off-diagonal
metric structures. The present framework does not incor-
porate these features and should therefore be understood
as a quasi-static sector of a more general relativistic the-
ory.

\paragraph{Summary of geometric effects.}

The results obtained
here demonstrate that metric--measure geometry alone
can generate several nontrivial features:

\begin{itemize}

\item the emergence of preferred hypersurfaces selected
by measure-weighted variational principles,

\item the existence of extremality conditions involving
both extrinsic curvature and measure gradients,

\item the modification of scaling behavior and localiza-
tion properties through the deformation function
$f$.

\end{itemize}

The examples considered illustrate complementary as-
pects of this mechanism. In the Schwarzschild case, the
measure-weighted functional exhibits localization at a fi-
nite radius determined by the parameter $\tau$, while in the
Anti-de Sitter case the measure deformation modifies the
ultraviolet scaling behavior of the functional. These re-
sults indicate that infrared localization and ultraviolet
scaling can be encoded within a unified geometric frame-
work.

\paragraph{Stability and variational structure.}

The analysis
in this work is restricted to first-order variations. The
stability of the resulting extremal surfaces depends on
the second variation of the functional, which is expected
to involve a modified Jacobi operator incorporating drift
terms arising from the measure deformation. A detailed
stability analysis would depend on curvature and higher
derivatives of $f$ and is left for future investigation.

\paragraph{Limitations and outlook.}

Several limitations remain.
The function $f$ is prescribed rather than dynamically de-
termined, and the bulk contribution introduced in Sec.~V
is treated at an effective level rather than derived from an
underlying quantum field theory. Moreover, the present
framework does not incorporate relativistic geometric
flow evolution, nonholonomic constraints, or generalized
connections.

Future work should aim to extend the present con-
struction by:

\begin{itemize}

\item embedding the variational framework into relativis-
tic geometric flow theories,

\item establishing a more direct relation to $W$-entropy
and thermodynamic variables,

\item deriving the bulk contribution from quantum or
semiclassical field-theoretic considerations,

\item and formulating a fully covariant Lorentzian ver-
sion of the extremality condition.

\end{itemize}

In summary, the results suggest that measure deforma-
tion provides an independent geometric mechanism capa-
ble of generating extremal surface structures and modi-
fying scaling behavior. This indicates that certain struc-
tural features commonly associated with holographic and
semiclassical entropy constructions may admit a deeper
origin rooted in metric--measure geometry itself, although
a complete physical interpretation remains an open prob-
lem.

\appendix

\section{Variational Structure of Measure-Weighted Functionals}

\subsection{Purpose and scope}

This appendix provides a unified derivation of Euler--
Lagrange equations associated with measure-weighted
functionals. The curve-based functional considered be-
low serves as a one-dimensional analogue of the hyper-
surface and codimension-two variational problems devel-
oped in the main text. It illustrates, in a simplified set-
ting, how measure deformation modifies standard geo-
metric extremality conditions and clarifies the normaliza-
tion structure underlying the higher-dimensional cases.

\subsection{Measure-weighted energy functional for curves}

Let $(M,g,f)$ be a smooth metric--measure space with
measure $d\mu=e^{-f}dV_g$. For a smooth curve $\gamma:[0,L]\to
M$, we consider the functional

\begin{equation}
E_f[\gamma]
=
\frac{1}{2}
\int_0^L
e^{-f(\gamma(s))}
|\dot{\gamma}(s)|^2 ds.
\label{eq:A1}
\end{equation}

This functional generalizes the classical energy by in-
corporating the measure weight and is consistent with
weighted variational principles arising in metric--measure
geometry and diffusion-type systems.

\subsection{First variation}

Let $\gamma_\epsilon(s)$ be a smooth variation with fixed endpoints
and variation vector field

\begin{equation}
V(s)
=
\left.
\frac{\partial \gamma_\epsilon(s)}{\partial \epsilon}
\right|_{\epsilon=0}.
\label{eq:A2}
\end{equation}

Denoting $\dot{\gamma}=\partial_s\gamma$ and $\nabla_s:=\nabla_{\dot{\gamma}}$, a standard computa-
tion using metric compatibility and integration by parts
yields

\begin{equation}
\delta E_f
=
\int_0^L
e^{-f}
\left\langle
V,
\nabla_s\dot{\gamma}
-
\frac{1}{2}
|\dot{\gamma}|^2\nabla f
+
\langle\nabla f,\dot{\gamma}\rangle\dot{\gamma}
\right\rangle
ds,
\label{eq:A3}
\end{equation}

where boundary terms vanish due to fixed endpoints.

\subsection{Euler--Lagrange equation}

Stationarity with respect to arbitrary variations $V$ im-
plies the Euler--Lagrange equation

\begin{equation}
\nabla_{\dot{\gamma}}\dot{\gamma}
=
\frac{1}{2}
|\dot{\gamma}|^2\nabla f
-
\langle\nabla f,\dot{\gamma}\rangle\dot{\gamma}.
\label{eq:A4}
\end{equation}

This equation generalizes the standard geodesic equa-
tion by introducing a drift term determined by the gra-
dient of the measure deformation.

\subsection{Arc-length parametrization}

If $\gamma$ is parametrized by arc length $(|\dot{\gamma}|^2=1)$,
equation~\eqref{eq:A4} simplifies to

\begin{equation}
\nabla_{\dot{\gamma}}\dot{\gamma}
=
\frac{1}{2}\nabla f
-
\langle\nabla f,\dot{\gamma}\rangle\dot{\gamma}.
\label{eq:A5}
\end{equation}

In this form, the acceleration is given by the projection
of $\frac{1}{2}\nabla f$ onto the normal bundle of the curve.

\subsection{Geometric interpretation}

The above equation shows that measure-weighted ex-
tremals deviate from standard geodesics by a force term
generated by $\nabla f$. More precisely, the right-hand side cor-
responds to the orthogonal projection of $\frac{1}{2}\nabla f$ onto the
normal space of the curve, while the tangential compo-
nent is removed by reparametrization invariance.

This structure can be interpreted as the autoparallel
equation of a torsion-free, non-metric connection

\begin{equation}
\nabla^{(f)}_X Y
=
\nabla_X Y
-
\frac{1}{2}
\left(
\langle\nabla f,X\rangle Y
+
\langle\nabla f,Y\rangle X
-
\langle X,Y\rangle\nabla f
\right),
\label{eq:A6}
\end{equation}

which is naturally associated with metric--measure geom-
etry.

\subsection{Normalization and relation to higher
codimension}

The coefficient $\frac{1}{2}$ appearing in equation~\eqref{eq:A4} reflects
the choice of weight $e^{-f}$ in the one-dimensional func-
tional. More generally, for a weight $e^{-\alpha f}$, the Euler--
Lagrange equation takes the form

\begin{equation}
\nabla_{\dot{\gamma}}\dot{\gamma}
=
\alpha
\left(
\frac{1}{2}
|\dot{\gamma}|^2\nabla f
-
\langle\nabla f,\dot{\gamma}\rangle\dot{\gamma}
\right).
\label{eq:A7}
\end{equation}

This normalization is consistent with the higher-
dimensional variational problems considered in the main
text. In particular:

\begin{itemize}

\item for codimension-one hypersurfaces, $\alpha=1$ leads to
$H=\nabla_n f$,

\item for codimension-two surfaces, $\alpha=\frac{1}{2}$ yields
$\vec{H}=\frac{1}{2}(\nabla f)^\perp$.

\end{itemize}

Thus, the curve case provides a unifying perspective on
the normalization of measure-weighted extremality con-
ditions.

\subsection{Relation to higher-dimensional variational
problems}

The structure derived here is directly analogous to
the hypersurface and codimension-two extremality con-
ditions obtained in the main text. In all cases, measure-
weighted variational principles produce a balance equa-
tion of the schematic form

\begin{equation}
(\text{geometric curvature})
\sim
(\text{projection of }\nabla f).
\label{eq:A8}
\end{equation}

For curves, this appears as a modification of the
geodesic equation; for hypersurfaces, it leads to a
weighted mean curvature condition; and for codimension-
two surfaces, it yields the extremality condition involving
the mean curvature vector.

\subsection{Remark on relativistic extensions}

The derivation presented here is formulated in a Rie-
mannian setting. In relativistic generalizations, the role
of the covariant derivative and Laplace-type operators
is replaced by hyperbolic counterparts (e.g.\ d'Alembert-
type operators), and additional structures such as non-
holonomic splittings and generalized connections may be
required. Extending measure-weighted variational prin-
ciples to such settings remains an open direction.

\subsection{Summary}

This appendix establishes that measure deformation
induces a universal modification of variational equations
across different geometric dimensions. The curve case
provides a transparent and technically controlled illustra-
tion of the mechanism underlying the extremality condi-
tions studied in the main text.

\end{document}